\documentclass[aps,pra,10pt,a4paper,reprint,nofootinbib]{revtex4-2}
\usepackage[utf8]{inputenc}
\usepackage{amsfonts}
\usepackage{amssymb}
\usepackage{amsmath}
\usepackage{amsthm}
\usepackage{graphics}
\usepackage{graphicx}
\usepackage{braket}
\usepackage[colorlinks=true,linkcolor=blue,urlcolor=blue,citecolor=blue]{hyperref}
\usepackage{times,txfonts}
\usepackage{changes}
\newtheorem{theorem}{Theorem}
\newtheorem{corollary}{Corollary}
\newtheorem{definition}{Definition}
\newtheorem{proposition}{Proposition}

\newtheorem{question}{Question}

\theoremstyle{definition}
\newtheorem{example}{Example}

\begin{document}
\title{Local unambiguous unidentifiability, entanglement generation, and Hilbert space splitting}

\author{Saronath Halder}
\affiliation{Center for Theoretical Physics, Polish Academy of Sciences, Aleja Lotnik\'{o}w 32/46, 02-668 Warsaw, Poland}

\author{Remigiusz Augusiak}
\affiliation{Center for Theoretical Physics, Polish Academy of Sciences, Aleja Lotnik\'{o}w 32/46, 02-668 Warsaw, Poland}

\begin{abstract}
We consider collections of mixed states supported on mutually orthogonal subspaces whose rank add up to the total dimension of the underlying Hilbert space. We then ask whether it is possible to find such collections in which no state from the set can be unambiguously identified by local operations and classical communication (LOCC) with non-zero success probability. We show the necessary and sufficient condition for such a property to exist is that the states must be supported in entangled subspaces. In fact, the existence of such a set guarantees the existence of a type of entangling projective measurement other than rank one measurements and vice versa. This projective measurement can create entanglement from any product state picked from the same Hilbert space on which the measurement is applied. Here the form of the product state is not characterized. Ultimately, these sets or the measurements are associated with the splitting of a composite Hilbert space, i.e., the Hilbert space can be written as a direct sum of several entangled subspaces. We then characterize present sets (measurements) in terms of dimensional constraints, maximum-minimum cardinalities (outcomes), etc. The maximum cardinalities of the sets constitute a class of state discrimination tasks where several stronger classes of measurements (like separable measurements, etc.) do not provide any advantage over LOCC. Finally, we discuss genuine local unambiguous unidentifiability and generation of genuine entanglement from completely product states. 
\end{abstract}
\maketitle

\section{Introduction}\label{sec1}
Distinguishing quantum states \cite{Chefles00, Barnett09, Bae15} is a key step in many quantum information processing protocols \cite{Nielsen00}. If the possible states of a quantum system are orthogonal then, in principle those states can be distinguished by a suitable global measurement. Eventually, one can identify the unknown state of the quantum system perfectly. However, the situation becomes complex when the quantum system is a composite one and no global measurement is allowed. In this scenario, the subsystems of the composite system are distributed among several spatially separated parties. Moreover, these parties are restricted to perform local operations only but to make strategies they can communicate classically. So, here the task of identifying the unknown state of the quantum system must be accomplished by local operations and classical communication (LOCC) for a given set of possible states. We say this task of distinguishing quantum states by LOCC as local state discrimination task. There are many instances where it is not possible to identify the state of a quantum system perfectly by LOCC even if the possible states are orthogonal \cite{Bennett99-1, Bennett99, Walgate00, Ghosh01, Walgate02, Ghosh02, DiVincenzo03, Horodecki03, Ghosh04, Fan04, Horodecki04, Nathanson05, Watrous05, Niset06, Hayashi06, Cohen08-1, Feng09, Bandyopadhyay10, Bandyopadhyay11, Yu12, Zhang14, Chitambar14, Zhang15, Xu16, Halder18, Halder19, Cohen23}. It is important to mention here that apart from studying nonlocal properties of composite quantum systems \cite{Peres91, Bennett99-1, Walgate02, Horodecki03, Bandyopadhyay11, Chitambar13, Halder19}, local state discrimination tasks also find application in data hiding \cite{Terhal01, Eggeling02, DiVincenzo02, Lami18, Lami21, Ha24, Ha25} and secret sharing \cite{Markham08, Rahaman15}.

When the states of a given set cannot be distinguished perfectly, we often think about probabilistic distinguishability. In the probabilistic regime, there are two standard settings to consider: (i) minimum-error setting and (ii) unambiguous setting. For details, one can have a look into Refs.~\cite{Chefles00, Barnett09, Bae15}. While in the first method non-zero error probability is allowed, the second method is the error-free case with a provision of an inconclusive outcome. In our work, we focus on the unambiguous setting under LOCC, see Refs.~\cite{Chefles04, Duan07, Bandyopadhyay09-1, Chitambar14, Bandyopadhyay16, Zhang20, Halder22, Lugli21, Ha21, Cohen15}. For unambiguous distinguishability of a set of quantum states by LOCC, it is necessary that all of the states are unambiguously locally identifiable with non-zero success probability. On the other hand, if any state of a given set is unambiguously locally unidentifiable then, the set is called unambiguously locally indistinguishable. The necessary and sufficient condition for unambiguous local identification can be found in \cite{Chefles04}. In particular, it was proved in \cite{Chefles04} that to identify a state unambiguously under LOCC from a given set, it is necessary and sufficient to find a product state which has non-zero overlap with the state to be identified. But this product state must have zero overlap with the other states of the given set. Nevertheless, here we introduce a strong notion of indistinguishability for a class of orthogonal mixed states. We ask if it is possible to find a set of orthogonal mixed states with the property that {\it no state of the set is unambiguously locally identifiable with non-zero success probability}\footnote{Though we do not mention it everywhere but for this paper the following is true. When we say unambiguous identification of a state is possible, we basically mean $0<p<1$ and when it is not possible $p=0$, where $p$ is the corresponding success probability.}. Clearly, if such a set exists, then the state in which the given quantum system is prepared, cannot be learned with non-zero success probability by LOCC error-freely. Thus, a set with this notion of local indistinguishability may have cryptographic importance. However, while this property guarantees local unambiguous indistinguishability, it is not a necessary condition for such variant of indistinguishability. If we stick to pure states, then any orthonormal basis composed only of entangled states has the aforesaid property. However, a set of orthogonal mixed entangled states whose ranks add up to the total dimension of the considered Hilbert space, may not possess this property. 

There is a possibility that the present states have connection with some interesting measurements. This is because here we consider only those orthogonal mixed states that cover the whole Hilbert space, i.e., the sum of the ranks of the states is equal to the total dimension of the given Hilbert space. Now this property of the mixed states might be useful to fulfill the completeness relation for a measurement. This is how the possibility of connection arises. However, here we are interested in global measurements which can produce entanglement from product states. Since, we are interested in post-measurement state, we consider projective measurements to generate entanglement because in this case the post-measurement states are well-defined. Moreover, even if we consider a measurement defined by a set of positive operator valued measure (POVM) elements, such a measurement can be treated as a projective measurement in an extended Hilbert space \cite{Nielsen00}. So, here the question of interest is to find a global projective measurement other than a rank one measurement which can produce entanglement from any product state. Here the product state belongs to the same Hilbert space on which the measurement is acting but the form of the product state is not characterized. 

We next discuss about entangled subspaces briefly. This is because the mixed states, we are talking about or the type of projective measurements, we want to construct, are associated with such subspaces. For bipartite systems, an entangled subspace is a subspace of a given Hilbert space where it is not possible to find any product state. In multipartite systems, there are at least two prominent types of entangled subspaces. One is completely entangled subspace (CES) \cite{Parthasarathy04}--this is a subspace of a multipartite Hilbert space where it is not possible to find any completely product state. The other one is the genuinely entangled subspace (GES) \cite{Dem18}--this is a subspace of a multipartite Hilbert space where it is not possible to find any biseparable state. For the constructions of these entangled subspaces and their applications in quantum information processing, one can have a look into \cite{Bennett99, Wallach02, Parthasarathy04, Bhat06, Cubitt08, Walgate08, Johnston13, Dem18, Dem19, Agrawal19, Dem20, Lovitz22, Dem24} and the references therein. In this context, we mention that in Ref.~\cite{Augusiak11} a composite Hilbert space was expressed as a direct sum of two entangled subspaces and it was shown how this type of decomposition can be connected to a class of witness operators that are not optimal. See also \cite{Kye12} in this regard.

However, in our work, we show that the existence of the collections of mutually orthogonal mixed states with the above property guarantees the existence of global projective measurements which can create entanglement form product states and vice versa. In fact, this connection is through the splitting of a composite Hilbert space into several non-overlapping entangled subspaces. In other words, if we take the direct sum of these entangled subspaces then, the given composite Hilbert space is obtained. So, through our work we elucidate the connection of several topics: local unambiguous unidentifiability, entanglement generation from any product state via global projective measurements, and entangled subspaces. Notice that here the concept of splitting a composite Hilbert space into several non-overlapping entangled subspaces is a more involved problem than just constructing a single entangled subspace in the Hilbert space. In this way, apart from introducing a strong notion of local indistinguishability and establishing connections among several novel concepts, we also find new applications of already existing concepts.

We now provide the main contributions of this paper. 

\begin{itemize}
\item We introduce a notion of indistinguishability for collections of orthogonal mixed states: {\it no state from a given set of orthogonal mixed states can be unambiguously identified by LOCC with non-zero success probability}. 

\item We show that `entanglement in the mixed states' is not sufficient for the existence of the aforesaid notion of indistinguishability. Actually, the necessary and sufficient condition for this indistinguishability is that the mixed states must be supported in the entangled subspaces. This relates to the possibility of splitting a given multipartite Hilbert space into several non-overlapping entangled subspaces.

\item Then, we show that the existence of such splitting of the Hilbert space into direct sums of entangled subspaces implies the existence of global measurements which can create entanglement from unknown product states and vice versa. 

\item We further characterize present sets (measurements) in terms of dimensional constraints, maximum-minimum cardinalities (outcomes), etc. For example, it is impossible to construct such sets in a two-qubit system.

\item It is shown that the maximum cardinalities for the present sets constitute a class of state discrimination tasks for which several stronger classes of measurements (such as separable measurements, etc.) do not provide any advantage over LOCC.

\item Additionally, we introduce an elimination game for understanding the present sets in a better way.

\item For multipartite systems, we present genuine local unambiguous unidentifiability and generation of genuine entanglement from completely product states. These are done by splitting the Hilbert space into several non-overlapping genuinely entangled subspaces. 

\item We construct explicit examples to establish our claims.
\end{itemize}
 
In the following we first discuss the necessary assumptions. Then, we present the bipartite and the multipartite results one by one, addressing the main questions more precisely. Finally, we provide the conclusion.

\section{Assumptions}\label{sec2}
We consider quantum systems, associated with the Hilbert space, $\mathcal{H} = \mathbb{C}^{d_1}\otimes\mathbb{C}^{d_2}\otimes\cdots\otimes\mathbb{C}^{d_m}$, $\forall~i=1,\dots,m$, $d_i$ is the dimension of the $i^{th}$ subsystem, it is finite, and $d_i \geq 2$. Here $m$ is the number of parties and each party is holding only one subsystem. For bipartite systems, $m = 2$ and $\mathcal{H} = \mathbb{C}^{d_1}\otimes\mathbb{C}^{d_2}$, otherwise, $m>2$ for multipartite systems. Then, we consider a set of orthogonal quantum states $\mathbb{S}$, where all of the states are mixed. Mathematically, we can say the following.
\begin{equation}
\begin{array}{l}
\mathbb{S} \equiv \{\rho_1, \rho_2, \dots, \rho_n\},\\[1 ex]
\mbox{Tr}(\rho_i^2)<1,~ \mbox{Tr}[\rho_i\rho_j]=\delta_{ij},~ \forall i,j = 1,\dots,n,
\end{array}
\end{equation}
where `Tr' is the standard trace operation for matrices and $\delta_{ij}$ is the well-known Kronecker delta function. We also assume that the set of mixed states that we have considered, has a particular property: $\sum_{i=1}^n \mbox{rk}(\rho_i) = d_1d_2$ for bipartite systems and $\sum_{i=1}^n \mbox{rk}(\rho_i) = d_1d_2\dots d_m$ for multipartite systems, where $\mbox{rk}(\rho_i)$ is the rank of $\rho_i$ and it is $\geq2$. So, the supports of the mixed states are not overlapping and direct sum of these supports produce the whole composite Hilbert space, $\mathcal{H}$. We note that in a given set the mixed states are equally probable. These assumptions are applicable for all of the following results. 

\section{Bipartite systems}\label{sec3}
Now for bipartite mixed states, we want to explore their local indistinguishability property under unambiguous setting. However, when we wish to unambiguously distinguish states by LOCC, it is necessary to examine if each state of the given set is locally unambiguously identifiable. In this context, we ask the following question.

\begin{question}\label{qn1}
Is it possible to construct a set of orthogonal mixed states such that none of these states can be locally unambiguously identified with non-zero success probability?
\end{question}

\noindent
Here we try to solve this question under the assumptions mentioned in Sec.~\ref{sec2}. If we consider a set of states such that all of the states can be unambiguously identified by LOCC with some non-zero success probability, then we say that the given set is unambiguously locally distinguishable. On the other hand, if at least one state of the set cannot be unambiguously identified by LOCC then the set is unambiguously locally indistinguishable. Stretching this indistinguishability, we want a set of mixed states where none of the states can be unambiguously identified by LOCC with non-zero success probability. We say this indistinguishability property as `{\bf Property 1}'.

Let us now define the following sets: (i) $\mathcal{S}_1$ which represents the collection of all sets for which the states cannot be perfectly distinguished by LOCC, (ii) $\mathcal{S}_2$ which represents the collection of all sets for which the states cannot be unambiguously distinguished by LOCC with non-zero success probability, and (iii) $\mathcal{S}_3$ which represents the collection of all sets for which none of the states can be unambiguously identified by LOCC. Then, we have the following:

\begin{equation}\label{eq1}
\mathcal{S}_3\subset\mathcal{S}_2\subset\mathcal{S}_1.
\end{equation}

\noindent
A positive answer to Question \ref{qn1} tells us that $\mathcal{S}_3$ is non-empty. But $\mathcal{S}_3$ can be empty for some particular Hilbert spaces Clearly, answering Question \ref{qn1} is not easy because there can be different constraints. One such constraint is given as the following.

\begin{proposition}\label{prop1}
For two qubits, all orthogonal mixed states can be unambiguously identified by LOCC, does not matter how entangled they are.
\end{proposition}

\begin{proof}
The above proposition is basically due to an existing result. It states that if we consider a two dimensional two-qubit subspace then such a subspace always contains a product state \cite{Sanpera98}. Now, let us assume that the given set of orthogonal mixed states is $\{\rho_1, \rho_2\}$. Since, the minimum rank of a mixed state is two, in case of two qubits a set at most contains two orthogonal mixed states. So, the supports of these states are basically two dimensional two-qubit subspaces and they contain at least one product state. In this way, it is always possible to have $\langle\phi_i|\rho_j|\phi_i\rangle=p_i\delta_{ij}$, $\forall i,j = 1,2$, where $\ket{\phi_i}$ are product states, contained in the supports of $\rho_i$ and $p_i$ are non-zero probabilities with which the states can be unambiguously identified by LOCC. In fact, this is known to be necessary and sufficient condition for unambiguous identification of states by LOCC \cite{Chefles04, Bandyopadhyay09-1}. These complete the proof.
\end{proof}

\begin{example}\label{ex1}
To constitute an example, we consider two orthogonal mixed states for a two-qubit system, $\rho_1$ = $(1-p)|\Phi^+\rangle\langle\Phi^+| + p|01\rangle\langle01|$ and $\rho_2$ = $(1-p^\prime)|\Phi^-\rangle\langle\Phi^-| + p^\prime|10\rangle\langle10|$. Here $p, p^\prime$ are non-zero probabilities and $|\Phi^{\pm}\rangle$ = $(1/\sqrt{2})(\ket{00}\pm\ket{11})$. Then, we consider local measurements on both qubits in $\{\ket{0}, \ket{1}\}$ basis. For both measurements, if the outcomes are different, then the states can be distinguished unambiguously though they are both entangled. When the outcomes are same, the situation is inconclusive. Clearly, entanglement within the states is not a sufficient condition for {\bf Property 1}.
\end{example}

Notice that the states $\{\rho_1,\rho_2\}$ cannot be perfectly distinguished by LOCC because the projection operators onto the supports of these states are entangled \cite{Chitambar14}. In this way, this example belongs to the collection $\mathcal{S}_1$ but not to $\mathcal{S}_2$ or $\mathcal{S}_3$. We will again discuss about the relation given in Eq.~(\ref{eq1}), in a later portion of this paper. However, from the above discussions, it is also clear that even if there exists a solution to Question \ref{qn1}, it cannot be found in the two-qubit Hilbert space. Now, before we discuss more about Question \ref{qn1}, we want to raise another question. Apparently, the second question does not have any connection with the first one. But we will prove that this is not the case. 

Let us now state the next question that we consider in our work. 

\begin{question}\label{qn2}
Is it possible to design a global projective measurement other than rank-1 measurements such that when it is applied on any product state, it always outputs an entangled state, provided that the product state is picked from the same Hilbert space on which the measurement is acting?
\end{question}

So, basically, we look for a set of projection operators $\{\Pi_1, \Pi_2,\dots, \Pi_n\}$, where $\sum_{i=1}^n\Pi_i=\mathbb{I}$, $\mathbb{I}$ is the identity operator acting on the associated Hilbert space, $\mathcal{H}$. Since, we are not interested in a rank-1 measurement, so, here $\mbox{rk}(\Pi_i)>1$ $\forall i = 1,\dots,n$. Basically, each $\Pi_i$ has some matrix representation, so here the ranks correspond to those matrices. Then, we consider the following black box scenario. We consider an arbitrary product state $\ket{\alpha}\ket{\beta}\in\mathcal{H}$, the form of which is not characterized. But the following is true 

\begin{equation}\label{eq2}
\frac{\Pi_i\ket{\alpha}\ket{\beta}}{\sqrt{\langle\alpha\beta|\Pi_i|\alpha\beta\rangle}} \neq \ket{\alpha^\prime}\ket{\beta^\prime}.
\end{equation}

\noindent
Clearly, the output is always entangled and here $\ket{\alpha}\ket{\beta}\equiv\ket{\alpha\beta}$. In comparison with a measurement in an entangled basis, the present measurement has a reduced number of measurement outcomes but still such a measurement can produce entanglement deterministically from a product state, no matter what is the form of the product state (see Fig.~\ref{fig1}). Roughly speaking, this scenario has similarity with a prepare and measure scenario \cite{Pawlowski11} where the preparation device is a black box but the measurement device is trusted. Interestingly, here to guarantee that the output state is entangled, we do not need to construct any conditional probabilities which is usual in a device-independent scenario. So, we do not need any independent and identically distributed copies, known to be difficult to achieve in experiments.

We say the property of the measurement which is mentioned in Question \ref{qn2} as `{\bf Property 2}'. Interestingly, to guarantee this property, it is not sufficient to consider entangled projectors\footnote{If the normalized version of a projection operator is an entangled quantum state, then the projection operator is said to be an `entangled projector'.}. To understand this, we go back to Example \ref{ex1}. We consider two entangled projectors onto the supports of $\rho_1$ and $\rho_2$. They constitute a two-outcome projective measurement on two qubits. However, this measurement cannot produce entanglement from the product states $\ket{01}$ or $\ket{10}$, but it can produce entanglement from $\ket{00}$ or $\ket{11}$. Thus, we argue that the {\bf Property 2} is a stronger notion than the notion of generating entanglement from some product states.  

\begin{figure}
\centering
\includegraphics[width=0.48\textwidth]{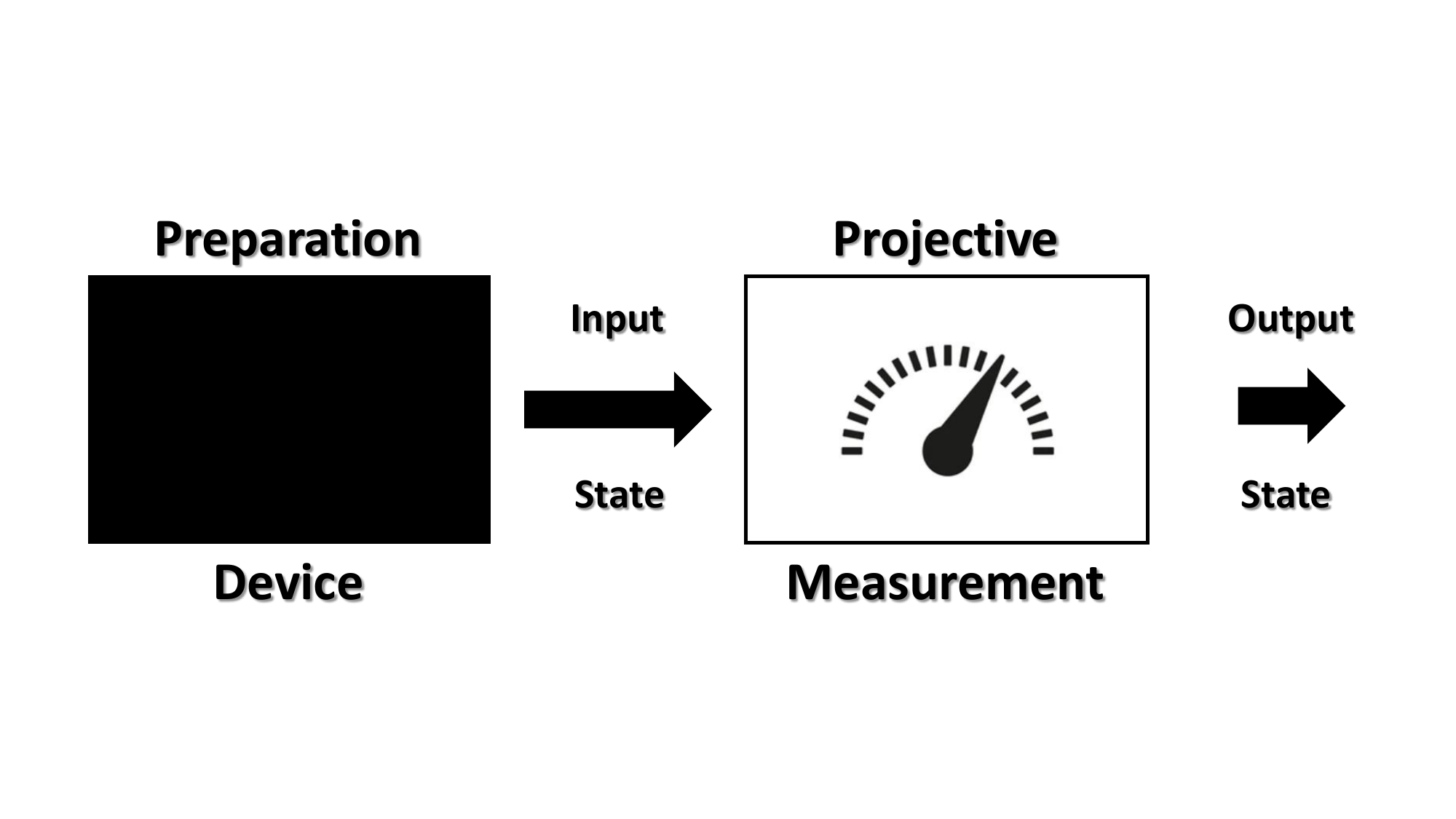}
\caption{The preparation device is a black box here. In particular, it is not known what is the form of the product state, produced by the device. Then, on this state a global projective measurement is performed. The measurement device is trusted. We want the output state to be entangled always.}\label{fig1}
\end{figure}

We are now ready to provide the following theorem. 

\begin{theorem}\label{thm1}
To ensure {\bf Property 1} ({\bf Property 2}), it is necessary and sufficient that the mixed states (projection operators) are supported in entangled subspaces. 
\end{theorem}

\begin{proof}
We start with a set of orthogonal mixed states \{$\rho_1$, $\rho_2$, \dots, $\rho_n$\}, for which $\sum_{i=1}^n\mbox{rk}(\rho_i)$ = $d_1d_2$, $d_1d_2$ is the total dimension of the considered Hilbert space. To identify any state $\rho_i$ unambiguously by LOCC, it is necessary and sufficient to find at least one product state $\ket{\phi}$, such that $\langle\phi|\rho_i|\phi\rangle = p>0$ and $\langle\phi|\rho_j|\phi\rangle=0$, $i\in\{1,2,\dots,n\}$ and $j\in\{1,2,\dots,n\}\setminus\{i\}$ \cite{Chefles04, Bandyopadhyay09-1}. Since, in our case, we have taken the condition $\sum_{i=1}^n\mbox{rk}(\rho_i)$ = $d_1d_2$, the product state $\ket{\phi}$ must be fully contained in the support of $\rho_i$. So, if we do not want the state $\rho_i$ to be locally unambiguously identified with non-zero success probability then it is necessary and sufficient that there is no such state $\ket{\phi}$ in the support $\rho_i$. Thus, to ensure {\bf Property 1}, it is necessary and sufficient that the states are supported in entangled subspaces.

Next, we consider a projective measurement defined by the projection operators $\{\Pi_1, \Pi_2,\dots, \Pi_n\}$. If these projection operators are supported in entangled subspaces then any product state $\ket{\alpha}\ket{\beta}$ is projected onto an entangled subspace after the application of a projection operator on it. So, the sufficient condition is straightforward. Now, for the necessary condition we first assume that there is a state $\ket{\alpha}\ket{\beta}$ contained in the support of any projection operator $\Pi_i$. We also assume an orthonormal basis $\{\ket{\psi_j}\}$ for the support of $\Pi_i$. So, $\Pi_i$ can be written as $\Pi_i = \sum_j|\psi_j\rangle\langle\psi_j|$. Since, $\ket{\alpha}\ket{\beta}$ is contained in a space for which the states $\{\ket{\psi_j}\}$ form a basis, $\ket{\alpha}\ket{\beta}$ should be written as linear combination of the states $\{\ket{\psi_j}\}$. So, $\ket{\alpha}\ket{\beta} = \sum_ja_j\ket{\psi_j}$, where $a_j$ are some complex numbers such that $\sum_j|a_j|^2 = 1$. In this way, if we apply $\Pi_i = \sum_j|\psi_j\rangle\langle\psi_j|$ on $\ket{\alpha}\ket{\beta} = \sum_ja_j\ket{\psi_j}$, then we do not obtain any entangled state. Clearly, to ensure {\bf Property 2}, it is necessary that $\ket{\alpha}\ket{\beta}$ is not contained in the support of any $\Pi_i$. In other words, it is necessary that $\{\Pi_i\}$ are supported in entangled subspaces.
\end{proof}

\begin{corollary}\label{coro1}
The existence of one property guarantees the existence of the other property and they are connected through the splitting of a composite Hilbert space into several non-overlapping entangled subspaces.
\end{corollary}

The necessary and sufficient conditions in case of both properties are similar, i.e., the states (projection operators) must be supported in entangled subspaces. Furthermore, the sum of the ranks of the states (projection operators) is equal to the total dimension of the composite Hilbert space. So, for both properties, the Hilbert space is split into several non-overlapping entangled subspaces such that the direct sum of these subspaces produces the composite Hilbert space. This is how the existence of one property guarantees the existence of the other property.

\begin{corollary}\label{coro2}
From the above, it is clear that the maximum rank of a mixed state should be equal to the maximum dimension of an entangled subspace to ensure {\bf Property 1}. Since, in case of two qubits, the maximum dimension of an entangled subspace is one, it is always the case that two-qubit mixed states can be unambiguously identified by LOCC with some non-zero success probability.
\end{corollary}

Let us recall that the maximum dimension of an entangled subspace in $\mathcal{H} = \mathbb{C}^{d_1}\otimes\mathbb{C}^{d_2}$ is $(d_1-1)(d_2-1)$ \cite{Wallach02, Parthasarathy04, Johnston13}. However, in the following, we construct examples via which we give positive answers to the questions raised earlier. We also discuss about maximum and minimum cardinalities of the sets associated with Question \ref{qn1}. Here cardinality is basically the number of outcomes associated with Question \ref{qn2}.

\subsection{Minimum dimensional construction}\label{sub1sec3}
Since, in $\mathbb{C}^2\otimes\mathbb{C}^2$ it is not possible to construct mixed states with {\bf Property 1}, we focus on $\mathbb{C}^2\otimes\mathbb{C}^3$, the minimum dimensional case and provide a construction. But before we present it, we mention that in $\mathbb{C}^2\otimes\mathbb{C}^3$, the maximum dimension of an entangled subspace is two. Therefore, we can get only three mixed states with the desired properties. This is the only possible cardinality that one can get here.

\begin{example}\label{ex2}
We consider the following entangled states first.

\begin{equation}\label{eq3}
\begin{array}{c}
\ket{\psi_1} = \ket{01} + \ket{10}, ~~ \ket{\psi_3} = \ket{02} + \ket{11},\\[1.0 ex]

\ket{\psi_2} = \ket{00} + \ket{12}, ~~ \ket{\psi_4} = \ket{00} - \ket{12},\\[1.0 ex]

\ket{\psi_5} = \ket{01} - \ket{10},\\[1.0 ex]

\ket{\psi_6} = \ket{02} - \ket{11}.
\end{array}
\end{equation}
We have not considered any normalization coefficients here, however, considering them, one can think about putting some bounds on the entanglement contents of the entangled subspaces that we are going to construct. We now consider the pairs of states: $\{\ket{\psi_1},\ket{\psi_2}\}$, $\{\ket{\psi_3},\ket{\psi_4}\}$, and $\{\ket{\psi_5},\ket{\psi_6}\}$ to produce entangled subspaces and then mixed states $\{\rho_1,\rho_2,\rho_3\}$ supported on these subspaces. These pairs form entangled subspaces of dimension two. The proof of producing entangled subspaces follows from the fact that if one takes superposition of the states within a pair, then, it is not possible to produce product states. For example, consider the following:
\begin{equation*}
\begin{array}{l}
a(\ket{01} + \ket{10}) + b(\ket{00} + \ket{12}) =\\[0.5 ex] 

\hspace{1.0 in}\ket{0}(b\ket{0} + a\ket{1}) + \ket{1}(a\ket{0} + b\ket{2}),\\[0.5 ex]




\end{array}
\end{equation*}
where $a, b$ are arbitrary complex coefficients such that $|a|^2+|b|^2=1$ and $|a|,|b|>0$. Notice that after superposition only entangled states are produced for arbitrary values of $a,b$. So, it follows from Theorem \ref{thm1} that if we consider any mixed states $\{\rho_1, \rho_2, \rho_3\}$, supported in these entangled subspaces, then none of these states can be unambiguously identified by LOCC with non-zero success probability. Thereafter, we consider three projection operators $\Pi_1$, $\Pi_2$, and $\Pi_3$ on these entangled subspaces. They constitute a rank-2 projective measurement which can create entanglement from any product state $\ket{\alpha}\ket{\beta}\in\mathbb{C}^2\otimes\mathbb{C}^3$.
\end{example}

From this example, it is now clear that $\mathcal{S}_3$ is non-empty, see (\ref{eq1}). Moreover, if we consider a set of only two states $\{\rho^\prime, \rho_3\}$, where $\rho^\prime$ is a rank-4 state and it is produced by taking a convex combination of $\rho_1$ and $\rho_2$, then, this set does not have {\bf Property 1} but it is unambiguously locally indistinguishable. In this case, the reason is simply because $\rho^\prime$ is no longer supported in an entangled subspace but $\rho_3$ is. In this way, $\{\rho^\prime, \rho_3\}$ belongs to $\mathcal{S}_2$ and $\mathcal{S}_1$ but not to $\mathcal{S}_3$. So, now due to the existence of sets like $\{\rho_1, \rho_2, \rho_3\}$ of the above example, $\{\rho^\prime, \rho_3\}$, and the set of Example \ref{ex1}, the relation of (\ref{eq1}) is established.   

The indistinguishable set $\{\rho_1, \rho_2, \rho_3\}$ of Example \ref{ex2} is important for another reason. This can be understood from the following. For all sets of mixed states constructed here, we assume that $\sum_{i=1}^n \mbox{rk}(\rho_i) = d_1d_2$. But we now mention that even if we do not assume this, there is no other type of sets for mixed states which can ensure {\bf Property 1} in $\mathbb{C}^2\otimes\mathbb{C}^3$. This is because when $\sum_{i=1}^2 \mbox{rk}(\rho_i) < 6$, it is always possible to find a state $\rho_i$ such that the space, orthogonal to the support of $\rho_i$, is not an entangled subspace. Therefore, this subspace, orthogonal to the support of $\rho_i$, contains at least one product state. This results the unambiguous local identification of the state other than $\rho_i$.   

However, we always compare the properties of pure states and mixed states. Example \ref{ex2} is also important from this point of view. In the following, we present a couple of points comparing the local (in)distinguishability properties between pure and mixed states.

\begin{itemize}
\item First, consider that two orthogonal pure states can always be perfectly distinguished by LOCC. But if we consider mixed states then, they may not be perfectly distinguished by LOCC \cite{Bandyopadhyay11, Halder21}. See also Example \ref{ex1}.

\item Now, in case of three pure states, it is always possible to extract `which state information' unambiguously by LOCC with some non-zero success probability \cite{Bandyopadhyay09-1}. However, here we have constructed a set of three mixed states for which it is not possible to extract `which state information' unambiguously by LOCC with some non-zero success probability.
\end{itemize}

Clearly, mixed states may exhibit more local indistinguishability compared to pure states.

\subsection{Higher dimensional generalization and degeneracy}\label{sub2sec3}
In $\mathbb{C}^2\otimes\mathbb{C}^3$, the only possible cardinality is three. We denote this by $(2,2,2)$, i.e., there are three states, each of which has rank two. In $\mathbb{C}^2\otimes\mathbb{C}^4$ there are two possibilities: $(2,2,2,2)$ and $(3,3,2)$. Here $(2,2,2,2)$ means that there are four states, each of which has rank two. Similarly, $(3,3,2)$ means that there are three states, two of which have rank three and the other state has rank two. Clearly, here the maximum and the minimum cardinality are deviating from each other within the same Hilbert space. We say this as {\it degeneracy} of cardinality and the number of possible cardinalities as the degree of degeneracy. For $\mathbb{C}^2\otimes\mathbb{C}^4$, the degree of degeneracy is two. Now along with the constructions of $\mathbb{C}^2\otimes\mathbb{C}^4$, we develop methodologies via which one can construct such mixed states in higher dimensions. We first present the methodology corresponding to the maximum cardinality for even dimensional subsystems. So, we consider the following points.

\begin{itemize}
\item For even dimensional subsystems, it is always possible to break the whole Hilbert space into several $\mathbb{C}^2\otimes\mathbb{C}^2$ subspaces. 

\item Then, for each subspaces, we construct entangled basis like the Bell basis and we construct pairs of states from different subspaces. Each of these pairs spans a two-dimensional entangled subspace. 

\item In this way, we can construct $d_1d_2/2$ rank-2 orthogonal mixed states. None of them can be unambiguously identified by LOCC with non-zero success probability. Here both $d_1$ and $d_2$ are even and $d_1d_2/2$ is the upper bound on the cardinality.
\end{itemize}

Let us now illustrate these points with an example.

\begin{example}\label{ex3}
Let us consider the following basis in $\mathbb{C}^2\otimes\mathbb{C}^4$ such that first four states span a subspace and the last four states span another subspace. The basis is given by-

\begin{equation}\label{eq4}
\begin{array}{cc}
\ket{\psi_1} = \ket{00} + \ket{11}, & \ket{\psi_2} = \ket{00} - \ket{11},\\[1 ex]
\ket{\psi_3} = \ket{01} + \ket{10}, & \ket{\psi_4} = \ket{01} - \ket{10},\\[1 ex]
\ket{\psi_5} = \ket{02} + \ket{13}, & \ket{\psi_6} = \ket{02} - \ket{13},\\[1 ex]
\ket{\psi_7} = \ket{03} + \ket{12}, & \ket{\psi_8} = \ket{03} - \ket{12}.
\end{array}
\end{equation}

\noindent
Next, we consider the pairs of states: $\{\ket{\psi_1}, \ket{\psi_5}\}$, $\{\ket{\psi_2}, \ket{\psi_6}\}$, $\{\ket{\psi_3}, \ket{\psi_7}\}$, and $\{\ket{\psi_4}, \ket{\psi_8}\}$. It is easy to check that if we take superposition of the states within a pair then it is not possible to produce any product state. Therefore, the states within a pair span a two-dimensional entangled subspace. So, any rank-2 mixed states supported on these entangled subspaces, say, $\{\rho_1, \rho_2, \rho_3, \rho_4\}$ cannot be unambiguously identified by LOCC by Theorem \ref{thm1}. Here the upper bound of the cardinality is four. Furthermore, if we consider projection operators $\{\Pi_1, \Pi_2, \Pi_3, \Pi_4\}$ onto the entangled subspaces, then they constitute a projective measurement with {\bf Property 2}.
\end{example}

If both dimensions are not even then, we can think about another methodology, given in the following. This one is also useful if we want to achieve the lower bound, i.e., the minimum cardinality of such sets. 

\begin{itemize}
\item We consider the computational (product) basis $\{\ket{00}$, $\ket{01}$, \dots, $\ket{d_1-1~d_2-1}\}$. Then, we consider combinations of these product states to produce entangled states and thereby, a complete entangled basis.

\item While constructing the entangled basis, we avoid taking states which form a complete basis for a lower dimensional subspace. For example, if we consider states $\{\ket{0x}\pm\ket{1x^\prime}, \ket{0x^\prime}\pm\ket{1x}\}$, $x,x^\prime$ $\in$ $\{0,$ $1$, \dots, $d_2-1\}$, $x\neq x^\prime$ which span $\mathbb{C}^2\otimes\mathbb{C}^2$ subspace, then taking superposition of any two states one can produce a product state. So, if we take $\ket{00}\pm\ket{11}$, then we consider a different product state to superpose with $\ket{01}$, instead of taking $\ket{10}$.

\item The next step is to consider combinations of these entangled states. But we check each pair of entangled states if they produce any product state after superposition. If not, then we consider superposition of more states and then we again check if one can get a product state after superposition. If not, then we get an entangled subspace of certain dimension. 

\item Finally, we have to find several such entangled subspaces which must be non-overlapping and then, taking direct sum of these subspaces, the given Hilbert space can be produced.
\end{itemize}

Let us now illustrate these points with an example.   

\begin{example}\label{ex4}
We construct the entangled basis in $\mathbb{C}^2\otimes\mathbb{C}^4$ (obeying the above points), given by- 
\begin{equation}\label{eq5}
\begin{array}{cc}
\ket{\psi_1} = \ket{00} + \ket{12}, & \ket{\psi_2} = \ket{00} - \ket{12},\\[1 ex]
\ket{\psi_3} = \ket{01} + \ket{13}, & \ket{\psi_4} = \ket{01} - \ket{13},\\[1 ex]
\ket{\psi_5} = \ket{02} + \ket{11}, & \ket{\psi_6} = \ket{02} - \ket{11},\\[1 ex]
\ket{\psi_7} = \ket{03} + \ket{10}, & \ket{\psi_8} = \ket{03} - \ket{10}.
\end{array}
\end{equation}
Notice that here no four states are of the form $\{\ket{0x}\pm\ket{1x^\prime},~ \ket{0x^\prime}\pm\ket{1x}\}$, $x,x^\prime$ $\in$ $\{0,$ $1$, \dots, $d_2-1\}$, $x\neq x^\prime$. Next, we take the following combinations: $\{\ket{\psi_1}, \ket{\psi_3}, \ket{\psi_5}\}$, $\{\ket{\psi_4}, \ket{\psi_6}, \ket{\psi_7}\}$, and $\{\ket{\psi_2}, \ket{\psi_8}\}$. Within a combination of three or two states, we check if any superposition of the states can produce a product state. In our case, this is not happening. In this way, we get a splitting of the $\mathbb{C}^2\otimes\mathbb{C}^4$ into two three-dimensional entangled subspace and one two dimensional entangled subspace. So, we can construct three mixed entangled states $\{\rho_1,\rho_2,\rho_3\}$ of ranks $(3,3,2)$ supported in these subspaces and they can have {\bf Property 1} by Theorem \ref{thm1}. Likewise, if we consider projection operators $\{\Pi_1, \Pi_2, \Pi_3\}$ onto the constructed entangled subspaces then this constitutes the projective measurement with {\bf Property 2}. This is the minimum cardinality case for $\mathbb{C}^2\otimes\mathbb{C}^4$ and the measurement is with minimum number of outcomes.
\end{example}

We now consider subsystems with prime dimensions. 

\begin{example}\label{ex5}
We first consider an entangled basis in $\mathbb{C}^3\otimes\mathbb{C}^3$:
\begin{equation}\label{eq6}
\begin{array}{ll}
\ket{\psi_1} = \ket{00} + \ket{11} + \ket{22}, & \ket{\psi_2} = \ket{00} - \ket{11},\\[1 ex]

\ket{\psi_3} = \ket{00} + \ket{11} - 2\ket{22}, & \\[1 ex]

\ket{\psi_4} = \ket{01} + \ket{12}, & \ket{\psi_5} = \ket{01} - \ket{12},\\[1 ex]

\ket{\psi_6} = \ket{02} + \ket{20}, & \ket{\psi_7} = \ket{02} - \ket{20},\\[1 ex]

\ket{\psi_8} = \ket{10} + \ket{21}, & \ket{\psi_9} = \ket{10} - \ket{21}.
\end{array}
\end{equation}
Next, we consider the following combination of states: $\{\ket{\psi_1}$, $\ket{\psi_4}$, $\ket{\psi_6}\}$, $\{\ket{\psi_2}, \ket{\psi_7}, \ket{\psi_9}\}$, and $\{\ket{\psi_3}, \ket{\psi_5}, \ket{\psi_8}\}$. Each of these combinations produces a three dimensional entangled subspace. Clearly, any three mixed entangled states $\{\rho_1,\rho_2,\rho_3\}$ of ranks $(3,3,3)$, supported in these three entangled subspaces, have the {\bf Property 1} by Theorem \ref{thm1}. Likewise, if we consider projection operators $\{\Pi_1, \Pi_2, \Pi_3\}$ onto these entangled subspaces then they constitute the projective measurement with {\bf Property 2}.
\end{example}

This is the minimum cardinality case for $\mathbb{C}^3\otimes\mathbb{C}^3$ and likewise, the measurement is with the minimum number of outcomes. There can be another option for minimum cardinality in $\mathbb{C}^3\otimes\mathbb{C}^3$, for example, $(4,3,2)$.

If we want only two mixed states or a two-outcome projective measurement with the present properties ({\bf Property 1} and {\bf Property 2}) then, the minimum dimension is $\mathbb{C}^3\otimes\mathbb{C}^4$ because in this Hilbert space the maximum dimension of an entangled subspace is six and taking direct sum of two such subspaces, it is possible to produce the whole $\mathbb{C}^3\otimes\mathbb{C}^4$ Hilbert space. Such a construction is given in \cite{Augusiak11}. But this construction is presented in a completely different context. However, if one of the subsystems is a qubit, then, such a splitting is impossible. 

\subsection{A no-go result}\label{sub3sec3}
We start with a couple of definitions.

\begin{definition}{}\label{def1}
[Distillable entanglement] A mixed state $\rho$ is said to have distillable entanglement \cite{Bennett96, Bennett96-1, Horodecki98, Watrous04}, iff it is possible to extract pure entangled state from it under LOCC with some non-zero probability. Note that for the distillation process, it may require many identical copies of $\rho$.
\end{definition}

\begin{definition}{}\label{def2}
[Classes of measurements] We consider a measurement defined by a set of positive semi-definite operators $\{\pi_i\}$. Now, (i) the measurement is a separable measurement (SEP) iff $\pi_i = \mathcal{A}_i\otimes\mathcal{B}_i$, $\mathcal{A}_i,\mathcal{B}_i$ are local operators, (ii) the measurement is a positive under partial transpose (PPT) measurement iff all PPT quantum states remain PPT after the application of the measurement on them.
\end{definition}

It is well known that the following relation holds.

\begin{equation*}
\mbox{LOCC}\subset\mbox{SEP}\subset\mbox{PPT}\subset\mbox{ALL},
\end{equation*}

\noindent
where ALL is the set of all measurements. From this relation, it is quite clear that there are instances where PPT measurements can provide advantage over SEP or LOCC and the instances where SEP can provide advantage over LOCC \cite{Duan09, Yu14, Bandyopadhyay22}. We are now ready to present the following proposition.

\begin{proposition}\label{prop2}
The maximum cardinalities corresponding to the present sets constitute a class of state discrimination tasks where SEP and PPT measurements are not advantageous over LOCC. 
\end{proposition}

\begin{proof}
For a maximum cardinality, the mixed states can have rank two or three. Now suppose, we consider a separable measurement or a PPT measurement for which an operator $\pi_i$ unambiguously identifies a state $\rho_i$. Then, we need $\mbox{Tr}[\rho_j\pi_i]=p\delta_{ji}$, $p>0$. However, here it is not possible to find an operator $\pi_i$ which fits into a separable measurement or a PPT measurement. The reason is the following. The mixed states are supported in two-dimensional or three-dimensional entangled subspaces. So, $\pi_i$ is also fully contained into the support of $\rho_i$. This implies that it cannot have the form $\mathcal{A}_i\otimes\mathcal{B}_i$. Moreover, if the operator $\pi_i$ is applied on a PPT state, the output state resides into the support of $\rho_i$. This guarantees that the output state cannot be PPT. Because the states within the support of $\rho_i$ are at most rank-2 or rank-3 states and they must have distillable entanglement \cite{Horodecki03-1, Chen08}. Thus, there is no scope to show advantage of SEP and PPT measurements over LOCC.
\end{proof}

Next, we proceed to discuss about the connection of the present constructions with an {\it elimination game} under LOCC.

\subsection{Elimination game}\label{sub4sec3}
Let us now introduce an elimination game under LOCC. We consider a local measurement scheme and, based on each measurement outcome, attempt to determine that the quantum system is not prepared in a particular state or set of states. Here we are interested in unambiguous conclusions only. More precisely, we consider a given set $\{\rho_1,\rho_2,\dots\}$. Then, we consider a local measurement scheme defined by $\{\pi_1, \pi_2,\dots\}$. For local state elimination, one can expect the following relations to be satisfied: $\mbox{Tr}[\rho_i\pi_i]=0$ and $\mbox{Tr}[\rho_j\pi_i] \neq 0$ for $i\neq j$. Then, for the measurement outcome ``$i$'', one can unambiguously conclude that the given quantum system is not prepared in $\rho_i$. In case of state discrimination, we try to identify the state in which the quantum system is prepared. However, when state discrimination is not possible, we try to eliminate one or more states from the given set. Previously, state elimination problem was studied under orthogonality-preserving LOCC in \cite{Halder18}. However, here we consider any LOCC protocol, not just orthogonality-preserving LOCC and we want to examine if it is possible to eliminate state(s) corresponding to each measurement outcome. Let us take an example to understand this clearly.

For Example \ref{ex2}, the parties consider a measurement in a product basis $\{\ket{00}$, $\ket{01}$, $\ket{02}$, $\ket{10}$, $\ket{11}$, $\ket{12}\}$, known to be locally implementable. Then, corresponding to each measurement outcome, it is possible to eliminate one state from $\{\rho_1, \rho_2, \rho_3\}$ unambiguously. The entire list is given below.
\begin{equation}\label{eq7}
\begin{array}{cccccc}
\ket{00} &  \rightarrow & \rho_3, & \ket{01} &  \rightarrow & \rho_2, \\[0.5 ex]

\ket{02} &  \rightarrow & \rho_1, & \ket{10} &  \rightarrow & \rho_2, \\[0.5 ex]

\ket{11} &  \rightarrow & \rho_1, & \ket{12} &  \rightarrow & \rho_3, \\[0.5 ex]
\end{array}
\end{equation}
where the states $\ket{ij}$ (on the left side of the right arrows) correspond to measurement outcomes and the states $\rho_i$ (on the right side of the right arrows) are the states which get eliminated.

But it may not always possible to eliminate state(s) corresponding to each measurement outcome. To understand this, we consider Example \ref{ex5}, where if the measurement outcomes correspond to $\ket{00}$ or $\ket{11}$, then it is not possible to eliminate any of the mixed states. Here the key factor can be described as the following. Let us consider some product states. We say them as composition product states. We assume that with these product states, it is possible to construct certain orthogonal pure entangled states. We further assume each of these entangled states belongs to the composition of different mixed states. Next, we consider a measurement in a product basis which contains all of the composition product states. Now, if we obtain an outcome which corresponds to one of the composition product states, then it does not help in elimination. As this product state is present in all mixed states.  

A key point in this context is that if a set contains only two mixed states with {\bf Property 1}, then, unambiguous state elimination is obviously not possible. Because in this case, elimination of one state means identification of the other. 

\section{Multipartite systems}\label{sec4}
The bipartite concepts can be generalized to multipartite cases in at least two different ways. The first one corresponds to a situation in which a multipartite Hilbert space can be written as the direct sum of several completely entangled subspaces (CESs). This can be done with the help of bipartite constructions. For example, if we consider the constructions of $\mathbb{C}^2\otimes\mathbb{C}^4$ and consider the following mapping: $\{\ket{0}\rightarrow\ket{00}, \ket{1}\rightarrow\ket{01}, \ket{2}\rightarrow\ket{10}, \ket{3}\rightarrow\ket{11}\}$, then, we basically get examples for three-qubit systems. Thus, we obtain three-qubit mixed states such that no state from the set can be unambiguously identified by LOCC with non-zero success probability. These constructions also depict the structures of global projective measurements that can create multipartite entangled states starting from any three-qubit product state $\ket{\alpha}\ket{\beta}\ket{\gamma}$. 

In the second way of multipartite extension, a multipartite Hilbert space is expressed as a direct sum of several genuinely entangled subspaces (GESs). If we consider mixed states, supported in these GESs, then these mixed states have the {\bf Property 1} across every bipartition. Because the GESs are entangled subspaces across every bipartition, and then one can apply Theorem \ref{thm1}. We call this property {\it genuine local unambiguous unidentifiability}. We further consider projection operators onto these GESs which constitute a global projective measurement having {\bf Property 2} across every bipartition, i.e., it can generate genuine entanglement from completely product state. These introduce genuine multipartite extensions of Questions \ref{qn1} and \ref{qn2}. We mention that in \cite{Dem20}, the three-qubit Hilbert space is split into three orthogonal GESs. However, in the following we present an interesting example for four qubits.

\begin{example}\label{ex6}
We construct rank-2 mixed states $\{\rho_i\}_{i=1}^8$, picked from the subspaces spanned by $\{\ket{\psi_i}, \ket{\psi_{i+8}}\}$, $i = 1,\dots,8$. We now consider the following basis.
\begin{widetext}
\begin{equation}\label{eq8}
\begin{array}{cc}
\ket{\psi_1} = \ket{0000} + \ket{1111},~~ \ket{\psi_2} = \ket{0000} - \ket{1111}, & 

\ket{\psi_3} = \ket{0011} + \ket{1100},~~ \ket{\psi_4} = \ket{0011} - \ket{1100},\\[1 ex]

\ket{\psi_5} = \ket{0101} + \ket{1010},~~ \ket{\psi_6} = \ket{0101} - \ket{1010}, &

\ket{\psi_7} = \ket{0110} + \ket{1001},~~ \ket{\psi_8} = \ket{0110} - \ket{1001},\\[1 ex]

\ket{\psi_9} = \ket{0001} + \ket{0010} + \ket{0100} + \ket{1000}, & \ket{\psi_{10}} = \ket{0001} + \ket{0010} - \ket{0100} - \ket{1000},\\[1 ex]

\ket{\psi_{11}} = \ket{0001} - \ket{0010} + \ket{0100} - \ket{1000}, & \ket{\psi_{12}} = \ket{0001} - \ket{0010} - \ket{0100} + \ket{1000},\\[1 ex]

\ket{\psi_{13}} = \ket{1110} + \ket{1101} + \ket{1011} + \ket{0111}, & \ket{\psi_{14}} = \ket{1110} + \ket{1101} - \ket{1011} - \ket{0111},\\[1 ex]

\ket{\psi_{15}} = \ket{1110} - \ket{1101} + \ket{1011} - \ket{0111}, & \ket{\psi_{16}} = \ket{1110} - \ket{1101} - \ket{1011} + \ket{0111},\\[1 ex]
\end{array}
\end{equation}   
\end{widetext}
where $\{\ket{\psi_1},\dots,\ket{\psi_8}\}$ are GHZ states and $\{\ket{\psi_9}$, \dots, $\ket{\psi_{16}}\}$ are W-states \cite{Dur00}. It is interesting that for four qubits, we construct a genuinely entangled basis where half of the states are GHZ and the remaining states are W states. Now, if we construct a two dimensional subspace spanned by $\ket{\psi_i}$ and $\ket{\psi_{i+8}}$ for any $i=1,\dots,8$, then the space is genuinely entangled. One may check that $a\ket{\psi_i} + b\ket{\psi_{i+8}}$ is entangled across every bipartition, where $a,b$ are some complex numbers and $|a|^2 + |b|^2=1$. We next consider the mixed states $\{\rho_i\}_{i=1}^8$. These states cannot be unambiguously identified with non-zero success probability by LOCC. More importantly, this unidentifiability property is preserved across every bipartition. Thus, we obtain a genuine version of {\bf Property 1}. Now, if we consider projection operators $\{\Pi_i\}_{i=1}^8$ onto the supports of $\{\rho_i\}_{i=1}^8$, then these operators constitute an eight-outcome rank-2 projective measurement which can produce genuinely entangled state from any four-qubit completely product state. In this way, we obtain a genuine version of {\bf Property 2}. This is the maximum cardinality (number of outcomes) for a set of four-qubit mixed states (a measurement on four qubits) with the desired property. 
\end{example}

In the above example, $\{\rho_i\}_{i=1}^8$ are supported in two-dimensional GESs. Therefore, one can think about applying Proposition \ref{prop2} to demonstrate no advantage in any bipartition. 

\section{Conclusion}\label{sec5}
In this work, we have introduced a strong notion of local indistinguishability for a class of orthogonal mixed states. The necessary and sufficient condition for such a notion to exist is connected with the concept of entangled subspaces. In fact, we have considered a more involved problem. We have constructed entangled subspaces such that direct sum of these subspaces produce the given composite Hilbert space. Then, we have constructed global projective measurements, other than rank-one measurements, which can create entanglement from any product state belonging to the same Hilbert space on which the measurement is applied. We have shown that the existence of the present notion of local indistinguishability guarantees the existence of the aforesaid measurements and vice versa. We have presented several examples to characterize these problems. Finally, we have discussed possible multipartite extensions of these problems. 

{\it Acknowledgment}.~Funded by the European Union under Horizon Europe (grant agreement no.~101080086). Views and opinions expressed are however those of the author(s) only and do not necessarily reflect those of the European Union or the European Commission. Neither the European Union nor the granting authority can be held responsible for them. R. A. is also supported by the Polish National Science Center through the SONATA BIS Grant No. 2019/34/E/ST2/00369.

\bibliographystyle{apsrev4-2}
\bibliography{ref}
\end{document}